\documentclass[10pt]{article}

\newif\ifaa
\newif\ifappendix
\appendixfalse

\newif\ifpolygrow
\polygrowtrue

\usepackage{mypaper}

\usepackage{enumitem} 
\usepackage{verbatim}
\def\llncs
\usepackage{algorithm}
\usepackage{algorithmic}
\newcommand{\nomore}[1]{\ensuremath{\le #1}}
\newcommand{\defeq}{\ensuremath{ := }}
\newcommand{\finInput}{\ensuremath{X}}

\title{The Polynomial Hierarchy does not collapse}
\titlerunning{PH does not collapse}
\makeatletter%
\ifdefined\lipics

\ccsdesc[500]{Theory of computation~Complexity classes}
\keywords{Complexity classes, Polynomial Hierarchy, Rice's Theorem, P versus NP}
\else
\ifllncs
\author{Reiner Czerwinski}
\institute{TU Berlin (Alumnus)}
\else
\ifsnjnl
\author*{\fnm{Reiner} \sur{Czerwinski}\email{Reiner.Czerwinski@posteo.de}}
\affil{No Institution}
\else
\author{Reiner Czerwinski\\TU Berlin (Alumnus)\\reiner.czerwinski@posteo.de}
\fi \fi
\titlerunning{PH does not collapse}
\fi
\iflipics
\author{Reiner Czerwinski}{TU Berlin (Alumnus)}{Reiner.Czerwinski@posteo.de}{https://orcid.org/0000-0002-4523-4420}{}
\fi
\authorrunning{R. Czerwinski}
\Copyright{Reiner Czerwinski}
\newcommand{\B}{\ensuremath{\{0,1\}}}
\ifllncs

\else

\fi

\begin{document}
\begin{headdata}

\begin{abstract}

      The arithmetical hierarchy ($\AH$) is similar to the polynomial hierarchy ($\PH$).
      Unlike the $\PH$, the $\AH$ does not collapse relative to any oracle.
      A $\Sigma_{k+1}$-complete language is computably enumerable (c.e.) relative to the
      $k$-th level of the $\AH$. So, given an oracle in the $k$-th level of the $\AH$, we require
      exhaustive search to recognize whether the input string is in the language.

 We can define a language in $\AH$ with a quantified formula and a computable relation.
 The encoding lengths of the bounded variables are unbounded.
 However, if we observe only finitely many input strings, then there exist integers
 such that the encoding lengths of the bounded variables can be limited to these integers
 without changing the result of the formula for each input string.

  By adding padding arguments of length equal to those of the bounded variables
and padding of length equal to the runtime required to compute the relation,
we construct an equivalent language in $\PH$.

An OTM operating on the language with padding arguments in $\PH$ with a padded oracle has the same
program flow
as it does on the original language with the unpadded oracle in the $\AH$.
The number of oracle queries does not change.
We still need brute-force search therefore, the required number of oracle queries
is exponential.
The $\PH$ does not collapse.

\ifllncs
\keywords{polynomial hierarchy, \and arithmetical hierarchy, \and complexity classes, \and black-box complexity, \and Turing degrees}
\fi
    \end{abstract}
\end{headdata}

\nocite{rogers1967theory}
\nocite{soare1987recursively}
\section{Introduction}
The polynomial hierarchy (\PH) and the arithmetical hierarchy (\AH)
are both defined by a computable relation combined with  alternating existential and
 universal quantifiers.
The number of alternating quantifiers denotes the level in the hierarchy. 
In the $\PH$,
unlike in the $\AH$,
the runtime and the lengths of the quantified variables are polynomially bounded.

In this paper, we will use query complexity to prove
that solving a problem in the $\AH$ with an OTM with an oracle one level below in the $\AH$ requires brute-force search.
Then we construct an equivalent language in the $\PH$ by using padding arguments. Thus, the $\PH$ problem requires
brute-force search with the padded oracle, too.

In Section~\ref{ahph}, we recall the definition of the $\AH$ and the $\PH$
and some properties of the $\PH$.


We analyze the query complexity of an $\AH$ language in $\Sigma_{k+1}$
when using a $\Pi_k$ oracle in Section~\ref{se:query}.
If we have an uncomputable oracle $A$
with a language $L$ c.e.\ in $A$, then
$L=\{x \mid \exists y\; f(x,y)\in A\}$, where $f$ is a computable
binary function.
If each input is contained in a finite set $\finInput$,
then there exists an integer $M$ such that
for all $x\in \finInput$:
\[x \in L \iff \exists y: |y|\le M \text{ and } f(x,y)\in A . \]

An oracle query is a Boolean function call.
If the oracle is not computable
then the OTM cannot
access the interior of the function, therefore, the query
is a black-box function from the OTM's perspective, and so, we require a brute force search.
We analyze the number of oracle queries
in Section~\ref{se:query}.
If we want to find exactly one value a search space of $N$ possible values,
then the expected number of oracle queries is $N/2 +1/2$~\cite[Chapter 9]{wegener2005complexity}.
Related to the encoding length of $N$, the number of queries grows exponentially.
By Yao's Minimax-Principle~\cite{Yao77}, the number of queries is also exponential
for an unknown distribution.

For an uncomputable language in $\AH$ defined with a quantified computable
relation, we can limit the encoding lengths of the bounded variables, whenever
we consider only a finite number of strings.
We prove this in Section~\ref{se:finiteInput}.

In Section~\ref{nocollapse}, we combine 
Section~\ref{se:finiteInput} with the oracle complexity results
of Section~\ref{se:query}.
For a language in $\AH$ with an uncomputable oracle one level below in the $\AH$,
we need brute force search. If we use the corresponding padded language
in $\PH$ with the oracle with paddings, then we need an equal number
of oracle queries.
We conclude that the $\PH$ does not collapse.

The $\AH$ does not collapse
relative to any oracle by Post's theorem~\cite{post1948degrees}.
However,

a proof for the $\PH$ must circumvent the relativization barrier.
On the one hand,
there are oracles relative to which the $\PH$ does not collapse~\cite{Yao1985SeparatingTP}. The $\PH$ does not collapse relative to random oracles~\cite{Hstad2017AnAD}.
A good overview can be found in~\cite{Rossman2015ComplexityTC}.
On the other hand, 
 for every $k\in\nat$, there is an oracle 
 for which the polynomial hierarchy collapses on exactly the $k$-th level relative to it~\cite{ko88}.

 Most proofs, especially when using oracles, relativizes.
 However,
 by Aaronson and Wigderson~\cite{algebrization},
 there is an exception:
          ``On the other hand, if we allow only polynomially-long queries,
then results based on padding -- for
example, $\P = \NP \Longrightarrow \EXP = \NEXP$ -- will generally fail to relativize.''

 In Section~\ref{se:rela} we argue
\ why the proof in this paper fulfills this non-relativizing condition.
\nocite{rogers1967theory}
\newcommand{\rlf}{\ensuremath{\tilde{R}_{L}}}
\newcommand{\msize}[2]{\ensuremath{\eta(#1,#2)}}
\newcommand{\defiff}{\ensuremath{\overset{\mathrm{def}}{\iff}}}
\newcommand{\defgl}{\ensuremath{\overset{\mathrm{def}}{=}}}
\section{Notations and Prerequisites}\label{ahph}
\subsection{Languages}
A language is a set of strings. A string is a finite sequence of symbols.
Without loss of generality we can choose $\B$ as the set of symbols. Thus, the strings are in $\B^*$.
The complement of a language $L$ is $\overline{L}=\bigl(\B^*\setminus L\bigr)$.

In this paper we
assume that the set of input strings $X$ is finite.
We analyze whether a string is in $\finInput\cap L$, or $\finInput\cap \overline{L}$.

\subsection{Relations}
If $A_1, \dots, A_n$ are sets,
then $R \subseteq A_1 \times \cdots \times A_n$ is an $n$-ary relation.
If $(a_1, \dots, a_n) \in  R$,
then we denote $R(a_1, \dots, a_n)$;
otherwise  $\lnot R(a_1, \dots, a_n)$.

The characteristic function of an $n$-ary relation $R$ is
\[\chi_R(x_1,\;\dots\;,x_n)=\begin{cases}
    1 &\text{if }R(x_1,\;\dots\;,x_n) \\
    0 &\text{if }\lnot R(x_1,\;\dots\;,x_n) \\
  \end{cases}.
\]

A relation is computable if a Turing machine can decide whether
a tuple is in the relation or not.

A relation $R$ is Turing recognizable if there
exists a TM $M$ with \[ R =  \{ (x_1,x_2,\;\dots\;,x_k) \mid \langle x_1,x_2,\;\dots\;,x_k \rangle \in L(M)  \}.\]

By the halting problem, not every Turing recognizable relation is computable.
However, we can construct a computable relation 
from the Turing recognizable relation
if we stop the computation of the corresponding TM
and encode the number of steps into the last variable of the relation.
\begin{definition}\label{relcomb}
  Let $k>1$.
  If $R$ is a $k$-ary relation that is recognized by a TM $M_R$,
  and $Q$ is a quantifier,
  then $R_{Q}$ is a relation with the characteristic function
  \[
    \chi_{R_{Q}}\left(x_1,\dots,\langle 1^t,x_k\rangle\right) =
    \begin{cases}
      1 & \text{ if TM }M_R \text{ accepts }\langle x,y_1,\;\dots\;,y_k \rangle\text{ within }t\text{ steps}  \\
      0 & \text{ if TM }M_R \text{ rejects }\langle x,y_1,\;\dots\;,y_k \rangle\text{ within }t\text{ steps}  \\
      d(Q)& \text{otherwise}\\
    \end{cases}
    \]
    where
    \[ d(Q)= \begin{cases}1& Q = \forall\\ 0 &Q = \exists. 
      \end{cases}
    \]
\end{definition}

\newcommand{\inb}[1][\ast]{\ensuremath{\in\B^{#1}}}
\subsection{Arithmetical Hierarchy}
The $\AH$ was independently introduced by  Kleene~\cite{KleeneAH} and Mostowski~\cite{Mostowski1947}
in the 1940s.

A language $L$ is in $\Sigma_k$ if there exists a computable relation
$R_L$ with
\begin{equation}\label{defAH}
  \begin{split}x \in L \iff
  \exists y_1\in\B^*\; \forall y_2\in\B^* \dots Q_k y_k\in\B^* :\\
  R_L(x,y_1,y_2,\;\dots\;,y_k)
\end{split} 
\end{equation}
\[
\text{where }Q_i =\begin{cases} \exists \text{ if }i\text{ is odd} \\ \forall \text{ if }i\text{ is even} \end{cases}
\]

  The class $\Pi_k$ is defined similarly to $\Sigma_k$, except that
  the first of the alternating quantifiers is universal and not existential.
  A language $L$ is in $\Pi_k$, if $\overline{L}\in\Sigma_k$.

  The Arithmetical Hierarchy is defined as
  \[
\AH = \bigcup_{k\in\nat} \Sigma_k =\bigcup_{k\in\nat} \Pi_k .   
\]

\ifappendix \else
According to Kleene~\cite{KleeneAH} and Mostowski~\cite{Mostowski1947} we know that
$\Sigma_k \cup \Pi_k \subset \Sigma_{k+1} \cap \Pi_{k+1}$, and
$\Sigma_k \setminus \Pi_k \not=\emptyset$, and $\Pi_k \setminus \Sigma_k\not=\emptyset$.
 A language is in $\Sigma_{k+1}$ iff it is c.e.\ in $\Sigma_k$ and $\Pi_k$
 by Post's Theorem~\cite{post1948degrees}.
 Therefore, the $\AH$ does not collapse relative to any oracle.
\fi
 \begin{remark}
In the literature, the variables in the definition are often natural numbers.
This definition of $\AH$ is equivalent, since
there is a bijective computable function from $\nat$ to $\B^*$:
\begin{equation} \label{eq:bin}
\rho(n) = b_{k-1}b_{k-2}\dots b_0
\text{ where } n+1 = \sum_{j=0}^k b_j*2^j\text{ and }b_k=1.
\end{equation}
\end{remark}


\ifllncs
\newcommand{\orl}{\ensuremath{\mathcal{A}_{L}}}
\else
\newcommand{\orl}{\ensuremath{\mathfrak{A}_{L}}}
\fi

\begin{lemma}\label{ora}
  If $L$ is $\Sigma_{k+1}$, then there is a set $\orl$ in $\Pi_k$
  such that \[
   L = \{x | \exists y\in\B^*\; \langle x,y \rangle \in\orl\}.
  \]
\end{lemma}
\begin{proof}
  If $L\in\Sigma_{k+1}$ then there is a computable relation $R$
  such that
  \[L = \{ x \mid \exists y_1\;\forall y_2 \dots Q_{k+1}y_{k+1} \; R(x,y_1,y_2\dots,y_{k+1}) \}  \]
  with alternating existential and universal quantifiers.
  If we choose
  \[ \orl = \langle x,y \rangle \mid  \forall y_2 \dots Q_{k+1}y_{k+1} \; R(x,y,y_2\dots,y_{k+1}) \}   \]
  then $L = \{x | \exists y\in\B^*\; \langle x,y \rangle \in\orl\}$.
\end{proof}

\subsection{Polynomial Hierarchy}

\begin{definition}
  If $A$ is a set and $n>0$, then
  \[ A^{\nomore{n}} \defgl \bigcup_{k\le n} A^k\text{.} \]
\end{definition}
The length of a string $x\in A^*$ is $|x|=\min\{n \mid x\in A^{\nomore{n}}\}$.
The string $1^n$ is the unary encoding of $n$.


The polynomial hierarchy is $\PH = \bigcup_{k\ge 0} \Sigma_k^{\P} =\bigcup_{k\ge 0}  \Pi_k^{\P}$.

A language $L$ is in $\Sigma_k^{\P}$ or $\Pi_k^{\P}$ if there exists a relation $R_L$ computable in polynomial time and a polynomial $p$ such that
\[\begin{split} x \in L \Leftrightarrow Q_1y_1\in\B^{\le p(|x|)}\; \dots Q_{k}y_{k}\in\B^{\le p(|x|)}\;:\\ 
  R_L(x, y_1, \dots, y_k)\text{.}
\end{split}
\]
There are alternating existential and universal quantifiers,
\[\text{i.e.\ }
Q_i =\begin{cases} \exists \text{ if }i\text{ is odd} \\ \forall \text{ if }i\text{ is even} \end{cases}\text{ for }L\in\Sigma^{\P}_k \text{ and}\]
\[Q_i =\begin{cases} \forall \text{ if }i\text{ is odd} \\ \exists \text{ if }i\text{ is even} \end{cases}\text{ for }L\in\Pi_k^{\P}\text{.}\]

The $\PH$ collapses on the $k$-th level if and only if\\
$\Sigma^{\P}_k=\Pi^{\P}_k=\Sigma^{\P}_{k+1}=\Pi^{\P}_{k+1}$.
It was introduced by Stockmeyer~\cite{meyer1972equivalence,stockmeyer1976polynomial}.

The following theorem is an exercise in some books, e.g. Homework 7.11 in~\cite[p. 166]{Homer2011}. 
\newcommand{\nSk}{\overline{S}_k}
\begin{theorem}\label{pattern}
  For a positive integer $k$,
  a Turing recognizable $k$-ary relation $R$, and a quantifier $Q$, let
  the
  corresponding computable relation be $R_Q$ as defined in Definition~\ref{relcomb}.
  The set
  \[
    \begin{split}
      S_k = \{ \langle R , x, 1^{m_1},1^{m_2},\;\dots\;,1^{m_k}, 1^t\rangle \mid \\
      \exists y_1\in\B^{\nomore{m_1}} \forall y_2\in\B^{\nomore{m_2}} \dots
      Q_k\langle 1^t,y_k\rangle\in\B^{\nomore{m_k}} :\\
      R_{Q_k}( x,y_1,y_2,\;\dots\;,\langle 1^t,y_k\rangle)\} 
    \end{split}
  \]
  is $\Sigma^{\P}_k$-complete,
  and the set
  \[
    \begin{split}
      \nSk = \{ \langle R , x, 1^{m_1},1^{m_2},\;\dots\;,1^{m_k}, 1^t\rangle \mid \\
      \forall y_1\in\B^{\nomore{m_1}} \exists y_2\in\B^{\nomore{m_2}} \dots
      Q_k \langle 1^t,y_k\rangle\in\B^{\nomore{m_k}} :\\
      R_{Q_k}( x,y_1,y_2,\;\dots\;,\langle 1^t,y_k\rangle)\} 
    \end{split}
  \]
  is $\Pi^{\P}_k$-complete.
  \end{theorem}
\begin{proof}
  Obviously, $S_k$ is in $\Sigma_k^{\P}$.
  
If the language $L$ is in $\Sigma_k^{\P}$, then there exist polynomials $p$ and $q$, and an $(n+1)$-ary relation $R$ with
\[\begin{split}
    L = \bigl\{ x \mid \exists y_1\in\B^{\le p(|x|)} \forall y_2\B^{\le p(|x|)} \nlts \dots Q_k y_k\in \B^{\le p(|x|)} :\\R(x,y_1,y_2,\;\dots\;,y_k) \nlts\text{ and }R\text{ is computable within }q(|x|)\text{ steps }\bigr\}.
  \end{split}
\]
So, for all $x$, $\langle R,x,1^{p(|x|)},\;\dots\;,1^{p(|x|)},1^{q(|x|)} \rangle \in S_k$.
This implies that $L \le^{\P}_m S_k$.

The set $\nSk$ is in $\Pi^P_k$, and
\[ \langle R,x,1^{m_1},\dots,1^{m_k},1^t  \rangle\in\nSk \iff \langle \lnot R,x,1^{m_1},\dots,1^{m_k},1^t  \rangle\in\S_k. \]
Therefore, $\nSk$ is $\Pi^P_k$-complete.
\end{proof}

\section{Query Complexity in the Arithmetical Hierarchy}\label{se:query}
If a language $L$ is in $\Sigma_{k+1}$,
then it is c.e. relative to $\Pi_k$ therefore, there exists an oracle
$\orl$ such that
\[
  L = \{ x\mid \exists y\in\B^*\; \langle x,y \rangle \in \orl  \}.
  \]
  Assume $\orl$ is $\Pi_k$-complete.
  \newcommand{\orlb}{\ensuremath{\mathcal{B}}}
  By Myhill's isomorphy theorem~\cite{myhill1955},
  there exists a computable permutation $f$
  for any other $\Pi_k$-complete oracle $\orlb$,
  such that
  $x\in \orl \iff f(x)\in \orlb$.
  Using another $\Pi_k$-complete oracle only changes the order of
  the queries.
  
  The oracle $\orl$ is generally not computable. Thus, the oracle query is a Boolean black-box function, i.e.\ the OTM has no
  access to its internal working.
  Therefore, we need brute-force search,
even when the length of the bounded variable is limited.
\begin{lemma}\label{exhaus}
  Let $k\ge 1$.
  Suppose we have an OTM with a $\Pi_k$-complete oracle $\orl$.
  To decide whether an input is in the language
  \[
    L_m = \{ x\mid \exists y\in\B^{\le m}\; \langle x,y \rangle \in \orl  \} 
  \]
  with the OTM,
  we need $\Omega(2^m)$ oracle queries in the worst case.
\end{lemma}
\begin{proof}
  The oracle $\orl$ is not computable. Therefore, we require an oracle query to
  decide whether an element is contained in $\orl$. To disprove $x\in L_m$,
  we must check the oracle for every $y\in\B^{\le m}$. Therefore,
  we need $2^{m+1}-1$ queries.

\end{proof}
\begin{remark}
  The expected number of queries of the OTM in Lemma~\ref{exhaus} is $2^m/N$
  if there exist exactly $N$ values in the oracle.
\end{remark}

\section{$\AH$ on finite input sets}\label{se:finiteInput}

For any non-negative integer $n$, we denote by $X_n$ the set of strings
with length less than or equal to $n$.
If the alphabet is $\B$ then $X_n=\B^{\le n}$.

Every string has a finite length. Therefore,
for every string $x$ there exists an integer $n$ with $x\in X_n$.
In this section, we analyze languages in $\AH$
by observing only finitely many input strings.

For every $n$, the set $X_n$ has finitely many strings,
and the set of all strings is $\lim_{n\to\infty}X_n=\B^*$.
Thus, there is no loss of generality when we consider
only sets with finitely many strings.
\begin{theorem}\label{limit}
  Let $\orl$ be an oracle, and
  $L :=\{ x | \exists y \in \B^{\ast} \langle x,y\rangle \in \orl \}$
  be c.e.\ in $\orl$.
  For a finite set of 
  strings $\finInput$
  there exists an $M=\msize{\finInput}{\orl}$
  such that for all $x\in \finInput,m \ge M$:
  \[
      x\in L \iff \exists y\in\B^{\le m} \langle x,y\rangle\in\orl\text{.}
    \]  
   \end{theorem}
  \begin{proof}
    If $L=\{ x | \exists y \in \B^{\ast} \langle x,y\rangle \in \orl \}$,
    then for each $x\in L$ there exists a string $y\in\B^{\ast}$ with
    $\langle x,y\rangle \in \orl$, so we can define $m_x =\min\bigl(\{ |y|  \mid \langle x,y\rangle
    \in \orl \} \bigr)$.

    If $x\in L$, and $m\ge m_x$, then
    $\exists y\in\B^{\le m}\; \langle x,y\rangle \in \orl$.
    If $\finInput$ is a finite set of input strings, we can define
    \[M = \msize{\finInput}{\orl} \defeq \max\bigl(\{ m_x \mid x\in \finInput \land \ x\in L    \} \cup \{0\} \bigr) \text{.}\]
    Thus,
    \[\begin{split}
     \text{ for all }m \ge M, x\in \finInput  :\nlts x\in L \iff \exists y\inb[\le m]\langle x,y\rangle \in\orl\text{.}
   \end{split}\]

  \end{proof}

\newcommand{\quantor}{\ensuremath{\mathcal{Q}}}
\newcommand{\qsizeone}[1][X]{\ensuremath{\mu(#1;k,R_L)}}
\newcommand{\qsizeonepp}[1][\finInput]{\ensuremath{\mu(#1;k+1,R_L)}}
\begin{lemma}\label{finite}
Let $X$ be a finite set of input strings, and
  \[\begin{split}  L = \{ x \in\finInput \mid  Q_1y_1\in\B^{*}\dots Q_{k}y_k\in\B^{*}\;\nlts
    R_L(x ,y_1 ,\dots ,y_k) \}
  \end{split}
\]
where $Q_1,\;\dots\;, Q_k$ are alternating existential and universal quantifiers,
and $R_L$ is a computable relation.

There exist integers $m_1,\;\dots\;,m_k$
such that for $x\in\finInput$:
\[\begin{split} 
  x \in L \iff  Q_1y_1\in\B^{\nomore{m_1}}\dots Q_{k}y_k\in\B^{\nomore{m_k}}\nlts
    R_L(x ,y_1 ,\dots ,y_k).
\end{split}
\]

\end{lemma}
  \begin{proof}
    We prove this lemma by induction over $k$:
    
      For the base case $k=0$ there is nothing to do.

      For the induction step, we consider two cases.
      \begin{itemize}
        \item[(i)]
      If a language $L$ is $\Sigma_{k+1}$
      then there is a set $\orl$ in $\Pi_k$
      such that  \[
   L = \{x | \exists y\in\B^*\; \langle x,y \rangle \in\orl\}
  \]
  by Lemma~\ref{ora}.
  If $X$ is a finite set of strings
  then there is an integer $m$ such that for $x\in X$:
  \[ x \in L \iff \exists y\in\B^{\le m}\; \langle x,y \rangle \in\orl\]
  by Theorem~\ref{limit}.
  Every tuple $(x,y)$ is an element of the finite set $X \times \B^{\le m} $.
  Under the induction hypothesis,
  \newcommand{\orel}{\ensuremath{\tilde{\mathcal{R}}}}
  there is a computable relation $\orel$ with alternating quantifiers
  $Q_2=\forall,\;\dots\;,Q_{k+1}$,
   and integers $m_2,\;\dots\;,m_{k+1}$ such that
  \[\begin{split}
    \orl
    =\{\langle x,y\rangle \mid Q_2 y_2\in\B^{\le m_2}\dots Q_{k+1}y_{k+1}\in\B^{\le m_{k+1}}\\
    \orel(\langle x,y \rangle,y_2,\dots,y_{k+1})\}.
  \end{split}
\]
  If $R(x,y_1,y_2,\dots,y_{k+1}):=\orel(\langle x,y_1\rangle,y_2,\dots,y_{k+1})$
  then
  \begin{equation}\label{eq:boundedvars}
    \begin{split}
      L=\{ x \mid \exists y_1\in\B^{\le m_1}\;Q_2 y_2\in\B^{\le m_2}\dots\\ Q_{k+1}y_{k+1}\in\B^{\le m_{k+1}}
      R(x,y_1,y_2,\dots,y_{k+1})\}
    \end{split}
  \end{equation}
  for $x\in X$.
\item[(ii)]
  If the language $L$ is in $\Pi_{k+1}$, then $\overline{L}\in\Sigma_{k+1}$.
  As mentioned in Eq.~(\ref{eq:boundedvars}),
  there exist integers $m_1,\dots,m_{k+1}$ for a finite set of strings $X$
  such that for $x\in X$:
  \[\begin{split}
    x\in\overline{L} \iff \exists y_1\in\B^{\le m_1}\; \forall y_2\in\B^{\le m_2}
    \dots Q_{k+1} y_{k+1}\in\B^{\le m_{k+1}}\\ \lnot R(x,y_1,y_2,\dots,y_{k+1}).
  \end{split}
\]
  Therefore,
  \[
    \begin{split}
    x\in L \iff \forall y_1\in\B^{\le m_1}\; \exists y_2\in\B^{\le m_2}
    \dots \overline{Q}_{k+1} y_{k+1}\in\B^{\le m_{k+1}}\\ R(x,y_1,y_2,\dots,y_{k+1}).
  \end{split}
\]
\end{itemize} 
  \end{proof}
  \nocite{ClassRecTh}
  The arithmetical hierarchy does not collapse~\cite{KleeneAH,Mostowski1947}. Consequently, by Post's theorem~\cite{post1948degrees},
  $\Sigma_{k+1}$-complete problems have a higher Turing degree than
  $\Pi_k$-complete problems.
  This implies that the variable lengths $m_1,\;\dots\;,m_k$ 
  in Lemma~\ref{finite} can grow faster than any computable function.

  \begin{corollary}\label{limitall}
    Let $X$ be a finite set of strings, and \[L=\{x \mid Q_1 y_1\;Q_2 y_2\dots Q_k y_k : R(x,y_1,y_2,\;\dots\;,y_k) \} \]
    where $Q_1,Q_2,\;\dots\;,Q_k$ are alternating existential and universal quantifiers,
    and $R$ is a Turing recognizable relation.
    \begin{enumerate}
      \item \label{limitall:exist}
    Then there are integers $m_1,m_2,\;\dots\;,m_k$, and $t$, such that for $x\in X$:
    \begin{equation}\label{prob}\begin{split}
        x \in L \iff \\
        x \in \bigl\{ x \mid Q_1 y_1\in\B^{\le m_1}\;Q_2 y_2\in\B^{\le m_2}\dots Q_k y_k\in\B^{\le m_k} :\\
        R(x,y_1,y_2,\;\dots\;,y_k)\bigr \},
     \text{and }R\text{ is computable in }t\text{ steps.}
      \end{split}
    \end{equation}
    
  \item \label{limitall:undecide}
    For arbitrary positive integers $m_1,m_2,\;\dots\;,m_k$, and $t$,
    Property \ref{prob}
    is generally undecidable.
  \end{enumerate}
\end{corollary}
\begin{proof}
     First we prove Part~\ref{limitall:exist}. 
      By Lemma~\ref{finite}, there exist integers $m_1,\dots,m_k$
      such that we can bound the length of $y_i$ by $m_i$ for $i\in\{1,\dots,k\}$.
      In this case, we can observe only finitely many tuples in $X \times \B^{\le m_1} \times \cdots \times \B^{\le m_k}$.
      The relation $R$ is Turing recognizable,
      so, there exists a TM $M_R$ with
      \[ R(x,y_1,\dots,y_k) \iff \langle x,y_1,\dots,y_k \rangle \in L(M_R).  \]
      The set $X \times \B^{\le m_1} \times \cdots \times \B^{\le m_k}$ has
      finitely many elements,
      therefore, the runtime of $M_R$ for all accepted element in this set
      can be bounded by a positive integer $t$.

      To prove Part~\ref{limitall:undecide},
      we show that Property~\ref{prob} is non-trivial.
     If the language $L$ is empty the equivalence
    \begin{math}
     x \in L \iff \dots
   \end{math} is trivially true for all $x$.
   On the other hand, the required bounds $m_1,\dots,m_k$, and $t$ can grow faster than any computable function relative $|x|$.
   By Rice's theorem~\cite{RICE}, it is undecidable whether the bounds $m_1,\dots,m_k$, and $t$ are large enough to fulfill Property~\ref{prob}.
      
    \end{proof}

    
    \section{Relation between $\AH$ and $\PH$}\label{nocollapse}
    \newcommand{\Lpad}{\ensuremath{L_{\mathrm{pad}}}}

    For a language in $\AH$ with only finitely many input strings,
    we can construct an equivalent padded language in $\PH$
    by Corollary~\ref{limitall}.
    If the language in the $\AH$ is c.e.\ relative to an oracle, then
    the padded language requires exponentially many queries of the corresponding padded oracle.

 \begin{lemma}\label{main}
   Let $k$ be a positive integer.
            If one computes a language in $\Sigma_{k+1}^{\P}$ with a $\Pi_{k}^{\P}$-oracle,
            then in the worst case one needs exponentially many oracle  queries.
          \end{lemma}
          \begin{proof}
            We consider a set $\finInput$
            of only finitely many inputs
            and use padding arguments
            to construct a language in $\PH$ 
            from
            a language in
            $\AH$ via Lemma~\ref{finite}.

            If $L$ is a  $\Sigma_{k+1}$-complete language
            then there is a $\Pi_k$-oracle $\orl$
            such that
            \[ L = \{ x \mid \exists y\in\B^*\; \langle x,y \rangle \in\orl\} \]
            by Lemma~\ref{ora}.
            If we consider only a finite set of strings $\finInput$
            as inputs then, by Theorem~\ref{limit}, there is
            an integer $m=\msize{\finInput}{\orl}$, such that for
            $x\in\finInput$:
            \[x\in L \iff \exists y\in\B^{\le m}\; \langle x,y \rangle \in\orl. \]
            If $\orl$ is Turing-hard,
            then, as mentioned in Lemma~\ref{exhaus},
            we need 
            $\Omega(2^m)$ oracle queries.

            \newcommand{\orap}{\ensuremath{{\orl}_{pad}}}
            \renewcommand{\rlf}{\ensuremath{\tilde{\mathcal{R}}}}
            The set $\orl$ is in $\Pi_k$. Therefore, there exists
            a computable relation $\rlf$,
            and alternating universal and existential quantifiers
            $\quantor_1=\forall,\quantor_2,\;\dots\;,\quantor_k$
            such that
            \[
              \langle x,y\rangle \in \orl \iff
              \forall y_1\in\B^* \dots \quantor_ky_k\in\B^*\;
              \rlf(\langle x,y\rangle,y_1,\dots,y_k),  
            \]and
         \[\begin{split}
              x \in L \iff \exists y\in\B^*\;
              \forall y_1\in\B^* \dots \quantor_ky_k\in\B^*\;
              R(x,y,,y_1,\dots,y_k) \\
              \text{where }R(x,y,,y_1,\dots,y_k):=\rlf(\langle x,y\rangle,y_1,\dots,y_k).  
            \end{split}
          \]
            If $X$ is finite, then there is an integer
            $m=\msize{\finInput}{\orl}$, such that for
            $x\in\finInput$:
     \[
              x \in L \iff \exists y\in\B^{\le m}
              \forall y_1\in\B^* \dots \quantor_ky_k\in\B^*\;
                 R( x,y,y_1,\dots,y_k).  
            \] 
          The number of elements in $X\times\B^{\le m}$ is finite.
          Therefore, 
           there exist integers $m_1,\;\dots\;,m_k$,
          such that for  $(x,y)\in X\times\B^{\le m}$:
   \[
              \langle x,y\rangle \in \orl \iff
              \forall y_1\in\B^{\le m_1} \dots \quantor_ky_k\in\B^{\le m_k}\;
              \rlf(\langle x,y\rangle,y_1,\dots,y_k),  
            \]
          by Lemma~\ref{finite}.

          Recall the definition of the set $\overline{S}_k$ in
          Theorem~\ref{pattern}.
          We define
          \[
            L_{pad}:=\{\langle x,1^m,1^{m_1},\dots,1^{m_k},1^t\rangle \mid
            \langle R,x,1^m,1^{m_1},\dots,1^{m_k},1^t\rangle \in S_{k+1} \},
          \]and
%
          \[\orap := \{ \langle x,y,1^{m_1},\dots,1^{m_k},1^t\rangle \mid
            \langle \rlf,\langle x,y\rangle,1^{m_1},\dots,1^{m_k},1^t\rangle
            \in\overline{S}_k \}\]
            According to Corollary~\ref{limitall},
            there are integers $m_1,\;\dots\;,m_k,t$ such that\\
            $\forall x\in X, \forall y \in\B^{\le m}$
            \begin{equation}\label{prob:largeEnough}
              \begin{split}
                x\in L \iff \langle x,1^m,1^{m_1},\dots,1^{m_k},1^t\rangle \in L_{pad}  ,\text{ and}\\
                \langle x,y\rangle \in \orl \iff \langle x,y,1^{m_1},\dots,1^{m_k},1^t\rangle\in \orap .
                \end{split}
              \end{equation}
              \newcommand{\longprop}{long-padding-property}
              We denote the property in (\ref{prob:largeEnough}) as \longprop.
              
              Thus, $L_{pad}$ is the corresponding $\Sigma^P_{k+1}$-problem for the $\Sigma_{k+1}$-language $L$.
              The $\orap$-oracle is $\Pi^{\P}_k$ instead of $\Pi_k$ like $\orl$-oracle. 
              
 If the padding lengths $m_1,\;\dots\;,m_k$ and $t$ are large enough
 to fulfill the \longprop, we can substitute the $\orl$-oracle
          with the $\orap$-oracle. An oracle query with $\orap$ returns the same result as a query with $\orl$,
          and therefore, the behavior of the OTM does not change except that every oracle query
          requires additional padding arguments.
Therefore, searching with a padded oracle also requires exponentially many queries.

 To analyze the case in which the padding lengths $m_1,\;\dots\;,m_k,t$ are not
            large enough to fulfill the \longprop,
          we assume that an OTM with the $\orap$-oracle could run faster than
          brute-force search in this case.
            This exploit would allow us to decide the property.
            However, the \longprop{} is undecidable by Part~\ref{limitall:undecide}
             of Corollary~\ref{limitall}, 
            thus, the computation still requires exponentially many
            oracle queries with short paddings.

            

   \newcommand{\rlt}[1][t]{\ensuremath{R_{(L,#1)}}}
   \newcommand{\AAP}{\ensuremath{A_P}}
   \newcommand{\oAAP}{\ensuremath{\tilde{A}_P}}

   We can choose $\orl$
   such that $\orap$
   is $\Pi_k^{\P}$-complete.
   In this case, there exists a polynomial reduction
   for each set in $\Pi_k^{\P}$ to 
   $\orap$ and therefore,
   we need an exponential runtime for any $\Pi_k^{\P}$ oracle.
\end{proof}

          According to Lemma~\ref{main}, we need exponential runtime to compute
          a $\Sigma_{k+1}^{\P}$-complete language with a $\Pi_{k}^{\P}$-oracle.
          Therefore, $\Sigma_{k+1}^{\P} \neq \Pi_{k}^{\P}$ for every positive integer $k$. 
          
          \section{Notes on the Relativization Barrier}\label{se:rela}
          The polynomial hierarchy does not collapse. This implies 
          $\P\not=\NP$. Thus, the proof must not relativize~\cite{baker1975relativizations}.
          However,
          the proof in this paper is based on padding arguments, and the
          oracle-query lengths
          are bounded by the sum of the padding lengths.
          By Aaronson and Wigderson~\cite{algebrization},
          this proof method circumvents the relativization barrier\footnote{This phrase is self-plagiarism because it is used in many papers of the author since 2020}.
          %


\nocite{compweber}
\bibliography{lit}{}
\ifsnjnl
\bibliographystyle{sn-aps.bst}
\else
\bibliographystyle{plain}
\fi
\ifappendix
\else
\ifpolygrow
\appendix
\fi
\fi
\ifpolygrow
\section{Growing of the padding bounds }
Recall the definition of $X_n$, i.e.\ the strings of length
less or equal $n$, in Section~\ref{se:finiteInput}. If the language $L$
is in $\AH$,
then there exists a computable relation $R$ with
\begin{equation}\label{defL} L = \{ x \mid Q_1 y_1\in\B^* \;Q_2 y_2\in\B^* \; \dots Q_k y_k\in\B^*  \} 
\end{equation}
where $Q_1,Q_2,\dots,Q_k$ are alternating existential and universal quantifiers.
By Corollary~\ref{limitall},
there exist integers $m_1(n),m_2(n),\dots,m_k(n),$ and $t(n)$
such that for $x\in X_n$:
\[x\in L \iff Q_1 y_1\in B^{\le m_1(n)}\;Q_2 y_2\in B^{\le m_2(n)}\;\dots Q_k y_k\in B^{\le m_k(n)}\;
  R(x,y_1,y_2,\dots,y_k), \]
and $R$ is computable within $t(n)$ steps.
\begin{lemma}
  If $L$ is a language defined in Eq.~\ref{defL},
  and $m_i(n)$ is the bound of the encoding length of the variable $y_i$
  when the input is in $X_n$,
  and $t(n)$ is the maximum runtime for the relation $R$ in $X_n$,
  then in general it is undecidable whether $m_i(n)$ and $t(n)$ grow polynomially.
\end{lemma}
\begin{proof}
  If the language $L$ is empty,
  then the bounds $m_i(n)$ with $1\le i \le k$, and  $t(n)$ are in $O(1)$.
  On the other hand, each bound $m_i(n)$ with $i\le k$ and $t(n)$ can grow
  faster than any computable function.
  By Rice's theorem, it is generally undecidable
  whether $m_1(n),\dots,m_k(n),$ and $t(n)$ grow polyomially.
\end{proof}
Thus, no general algorithm can decide
whether a Language in $\AH$ is in $\PH$.
\fi
\end{document}